\documentclass[journal,12pt,twoside]{IEEEtran}

\usepackage{hyperref}

\usepackage{algorithm}
\usepackage{algpseudocode}

\usepackage{psfrag}   
\usepackage{cancel}

\usepackage{url}      

\usepackage{pgfplots}
\pgfplotsset{width=7cm,compat=1.3}

\usepackage{psfrag}
\usepackage{makecell}

\usepackage{mathtools}

\usepackage{mathtools}

\usepackage{xfrac}

\usepackage{multicol}

\usepackage {tikz}
\usetikzlibrary {positioning}

\usepackage{multirow}
\usepackage{latexsym}
\usepackage{pictex}
\usepackage{fancyvrb}
\usepackage{fancyhdr}
\usepackage{color}
\usepackage{rawfonts}
\usepackage{graphics}
\usepackage{graphicx}
\usepackage{caption}
\usepackage{subcaption}
\usepackage{amssymb,amsmath,amsthm,amsfonts}
\usepackage{ifthen}

\usepackage{bbm}
\usepackage{cite}
\usepackage{enumerate}
\usepackage{verbatim}

\usepackage{balance}

\newcommand{\nop}[1]{} 
\newcommand{\shorten}[1]{}
\newtheorem{proposition}{Proposition}
\newtheorem{theorem}{Theorem}
\newtheorem{definition}{Definition}

\newtheorem{lemma}{Lemma}
\newtheorem{remark}{Remark}
\newtheorem{corollary}{Corollary}

\newcommand{\signed}%
    {{\unskip\nobreak\hfill\penalty50
      \hskip2em\hbox{}\nobreak\hfil $\blacksquare$
      \parfillskip=0pt \finalhyphendemerits=0 \par}}

\DeclareMathOperator{\spn}{span}
\begin{document}

\title{Functional Broadcast Repair of Multiple Partial Failures in Wireless Distributed Storage Systems}

\author{
	\IEEEauthorblockN{Nitish Mital\IEEEauthorrefmark{1}, Katina Kralevska\IEEEauthorrefmark{2}, Cong Ling\IEEEauthorrefmark{1} and Deniz G\"{u}nd\"{u}z\IEEEauthorrefmark{1} \\ 
	\IEEEauthorblockA{\IEEEauthorrefmark{1}Department of Electrical \& Electronics Engineering, Imperial College London}\\
	\IEEEauthorblockA{\IEEEauthorrefmark{2}Department of Information Security and Communication Technology, NTNU, Norwegian University of Science and Technology}\\
	Email: \{n.mital,c.ling,d.gunduz\}@imperial.ac.uk, katinak@ntnu.no}
}

\maketitle

\begin{abstract}	
We consider a distributed storage system with $n$ nodes, where a user can recover the stored file from any $k$ nodes, and study the problem of repairing $r$ partially failed nodes. We consider \textit{broadcast repair}, that is, $d$ surviving nodes transmit broadcast messages on an error-free wireless channel to the $r$ nodes being repaired, which are then used, together with the surviving data in the local memories of the failed nodes, to recover the lost content. First, we derive the trade-off between the storage capacity and the repair bandwidth for partial repair of multiple failed nodes, based on the cut-set bound for information flow graphs. It is shown that utilizing the broadcast nature of the wireless medium and the surviving contents at the partially failed nodes reduces the repair bandwidth per node significantly. Then, we list a set of invariant conditions that are sufficient for a functional repair code to be feasible. We further propose a scheme for functional repair of multiple failed nodes that satisfies the invariant conditions with high probability, and its extension to the repair of partial failures. The performance of the proposed scheme meets the cut-set bound on all the points on the trade-off curve for all admissible parameters when $k$ is divisible by $r$, while employing linear subpacketization, which is an important practical consideration in the design of distributed storage codes. Unlike random linear codes, which are conventionally used for functional repair of failed nodes, the proposed repair scheme has lower overhead, lower input-output cost, and lower computational complexity during repair. 
\end{abstract}
\let\thefootnote\relax\footnotetext{This work was supported by the European Research Council (ERC) Starting Grant BEACON (grant agreement no. 677854), European Union`s H2020 research and innovation programme under the Marie Sklodowska-Curie Action SCAVENGE (grant agreement no. 675891), and UK EPSRC (EP/T023600/1) under the CHIST-ERA program (CHISTERA-18-SDCDN-001). This work has been partly presented in two conference versions \cite{8849496, 8613401}}

\IEEEpeerreviewmaketitle

\section{Introduction}
Caching popular contents closer to end-users, particularly in the available storage space at the wireless network edge, is attracting a lot of attention in recent years as a promising method to alleviate the increasing traffic on the backhaul links of wireless access points, and to improve the quality of service for end users, particularly by reducing the latency \cite{6495773, 8114221}, or the energy consumption \cite{7438743, 8374062}. The literature on distributed caching systems focuses mostly on the code design or the resource allocation for efficient storage of popular contents, assuming reliable cache nodes. However, storage devices are often unreliable and prone to failures; thus, efficient repair techniques that guarantee continuous data availability are essential for a successful implementation of distributed caching and content delivery techniques in practice \cite{7559801}. 

Maximum distance separable (MDS) codes are typically used for distributed caching of contents at multiple access points \cite{6495773, 8114221, 8613401}. MDS codes provide flexibility for storage so that users with different connectivity or mobility patterns can download a file from only a subset of the access points. In particular, an $(n, k)$ MDS code encodes a file of size $M$ bits by splitting it into $k$ equal-size packets and encoding them into $n$ packets, which are then stored at $n$ cache nodes. Each data packet consists of a set of symbols in a finite field $\mathbb{F}_q$. The original file can be reconstructed by accessing any $k$ out of $n$ packets from $k$ distinct access points. This property is called the \textit{data reconstruction property}. 
When some nodes partially or fully fail, or when content has to be moved to new cache nodes, their cache contents must be regenerated, either in the same failed node, or in new cache nodes, to which the content has to be moved, to be able to continue serving users with the same amount of redundancy. An important objective of edge caching in wireless networks is to reduce the backhaul link loads; therefore, we consider \textit{cache recovery at the edge}; that is, rather than updating the failed cache contents from a central server through backhaul links, the failed cache contents are regenerated with the help of surviving cache nodes. This is called the \textit{data regeneration property} of a code. 
The total amount of data transferred from the surviving nodes to repair the failed nodes is called the \textit{repair bandwidth}. A trivial way to repair failed nodes is to allow the nodes being repaired to connect to any $k$ surviving nodes, download the entire file, and recover the data that was stored originally. However, downloading the entire file to recover the data of a node that stores only a fraction of the file is not efficient. Conventional MDS codes treat the packets stored in a node as a single symbol belonging to the finite field $\mathbb{F}_q$. It can be shown that when the nodes are only permitted to perform linear operations over $\mathbb{F}_q$ when using conventional MDS codes, the total repair bandwidth cannot be smaller than the length of the entire file. Thus, conventional MDS codes have high storage efficiency, but their repair bandwidth is large when using naive repair mechanisms \cite{5550492}.

In contrast, \textit{regenerating codes} are codes that treat the data stored on each node as a vector of $S$ data packets. In particular, a file of size $M$ bits is split into $P$ data packets consisting of symbols in $\mathbb{F}_q$. A total of $S$ coded packets are stored at each node with a storage capacity of $\alpha$ bits. Linear operations over $\mathbb{F}_q$ in this case permit the transfer of a fraction of the data stored in a single node, thus achieving a lower repair bandwidth than conventional MDS codes with naive repair mechanisms. Similarly to existing literature \cite{6352912,10.1145/2897518.2897525,Tamo2017OptimalRO,8717612,7463553,7457282}, we refer to $S$ as subpacketization. Such codes, constructed over a vector alphabet, are called \textit{array codes} in the distributed storage literature. During repair, $r$ failed nodes are allowed to connect to $d$ surviving nodes, and download a total of $\gamma=d\beta$ bits to repair the lost contents, where $\beta$ is the number of bits transmitted by each of the $d$ surviving nodes that are connected to. Dimakis \textit{et al.} showed in  \cite{5550492} that there is a fundamental trade-off between the storage capacity $\alpha$ and the repair bandwidth $\gamma$ by mapping the repair problem in a distributed storage system to a multicasting network coding problem \cite{Gamal:2012:NIT:2181143} over an information flow graph \cite{850663}. The trade-off is shown to be characterized by the cut-set bound for linear network coding, and can be achieved with high probability using random linear coding \cite{Gamal:2012:NIT:2181143}. The analysis in \cite{5550492} focuses on a single node repair, that is, losing one node triggers the repair process. Existing literature has mainly focused on two extremes of this trade-off: the minimum-storage regenerating (MSR) point and the minimum-bandwidth regenerating (MBR) point. Apart from a low repair bandwidth, it is also desirable for a regenerating code to have low subpacketization. Subpacketization determines the smallest file size that can be encoded, and the number of operations required in the encoding and decoding processes. 

\subsection{Exact vs functional repair}
\indent Another terminology frequently found in the literature is that of \textit{exact repair} and \textit{functional repair}. Functional repair, first introduced in \cite{5550492}, refers to the repair process in which the replacement of a set of failed nodes is such that the data reconstruction property and data regeneration property are preserved in the resulting network of $n$ nodes, while the coded content in the replacement nodes may be different from the coded content in the original nodes. In contrast, exact repair, subsequently introduced and studied in \cite{6062413,5205898}, refers to the repair process in which the replacement nodes store exactly the same coded content as stored originally in the failed nodes. Exact repair is often preferred over functional repair because the former is predictable and does not require nodes to communicate their changing coding coefficients to all other nodes. However, functional repair is generally more flexible than exact repair, and is able to achieve lower repair bandwidth than exact repair \cite{6062413,6804941}.

\subsection{Related work} \label{sec:rw}

Existing literature on distributed storage codes consider the exact repair of a single node at a time \cite{Dimakis2007DeterministicRC,5503165,5961826,6352912,10.1145/2897518.2897525,7457282}. In \cite{5961826}, the product-matrix framework for constructing optimal repair MSR and MBR codes is proposed, which achieves the MBR point for all admissible parameters, and achieves the MSR point for parameters restricted to $d\geq 2k-2$. The product-matrix construction employs a subpacketization level that scales linearly in $n$, $k$ and $d$. Zig-zag codes and HashTag codes, proposed in \cite{6352912} and \cite{8025778} respectively, allow arbitrarily high code rates ($\sfrac{k}{n} \rightarrow 1$) for the MSR point, but require a subpacketization level that is exponential in $k$. Other MSR code constructions, including repair schemes for MDS codes like Reed Solomon (RS) codes, have been proposed that meet the cut-set bound, but their subpacketization level grows exponentially in $n$ \cite{10.1145/2897518.2897525,Tamo2017OptimalRO,8717612,7457282}. 

It was first observed in \cite{5402494} that multiple node repair, that is, when the repair process starts only after $r$ nodes fail, is more efficient in terms of the repair bandwidth per node, compared to repairing each node as it fails. The repair process may be initiated centrally by a controller, which monitors the state of each storage node, or in a decentralized fashion via periodic updates by every storage node of its state to the remaining nodes. In \cite{5978920} and \cite{6565355}, the authors introduce cooperative regenerating codes, which repair multiple failures cooperatively by allowing each of the $r$ nodes being repaired to collect data from any $n-r$ surviving nodes, and then to cooperate with the other $r-1$ nodes being repaired. In \cite{6620465}, the authors design a minimum storage cooperative regenerating (MSCR) code for $n=2k, d= 2k-2, r=2$, which they generalize to the repair of systematic nodes for $n=2k, d=n-r, 2\leq r\leq n-k$ in \cite{coop_repair}. An obvious drawback of these constructions is that they are restricted to a limited set of parameters. Ye and Barg \cite{8410934} give an explicit MSCR construction for all parameter values, that is, for $2\leq r\leq n-k, d> k$, but employ a subpacketization that is exponential in the parameters $n, k$ and $d$. The product-matrix construction of \cite{5961826} for a single node repair is extended to MSCR in \cite{8793154}, meeting the cut-set bound with a linearly scaling subpacketization, but for parameters restricted to $d\geq 2k-2-r$. In \cite{6566803}, the authors give an explicit construction of minimum bandwidth cooperative regenerating (MBCR) codes for all parameters $n, k, d, r$, employing only a linear subpacketization. 

\begin{table*}[!t] 
		\caption{Comparison of explicit constructions that achieve the cut-set bound.}
		\begin{center}
			\begin{tabular}{|c|c|c|}	
			\hline
			Code parameters &  Subpacketization & Ref. \\
				\hline
			$[n,k \rightarrow n, d, r]$ MSR array code & exponential in $k$ & \cite{6352912} \\
			\hline
			All feasible parameters, RS scalar code & $O(n^n)$ & \cite{Tamo2017OptimalRO} \\
			\hline 
			$[n,k,d\geq 2k-2, r=1] $ MSR array code & $d-k+1$ & \cite{5961826} \\
			\hline 
			$[n,k,d,r=1]$, MBR array code & $\frac{k}{2}(2d-k+1)$ & \cite{5961826} \\
			\hline 
			$[n=2k, k, d = n-r, 2\leq r \leq n-k]$, MSCR code & $d - k + r$ & \cite{6620465} \\
			\hline
		All feasible parameters, MSCR code & exponential in $n,k,d$ & \cite{8410934} \\
		\hline
		$[n,k,d\geq 2k-2-r, r]$, MSCR code & $d-k+r$ & \cite{8793154} \\
		\hline
		All feasible parameters, MBCR code & $\frac{k}{2}(2d-k+r)$ & \cite{6566803}\\
		\hline
			\end{tabular}
		\end{center}
	\end{table*}

Similarly to \cite{7000553,7459908}, we consider broadcast repair; that is, transmissions from each node are received in an error-free manner by all the other nodes. The storage-repair bandwidth trade-off for the repair of multiple fully failed nodes using broadcast repair is derived in \cite{7459908}. Additionally, we consider the \textit{partial repair} problem, studied in \cite{7000553}, in which a node loses only a part of its contents, and the remainder of the contents may be used along with the transmissions from the surviving nodes to repair the content, thus further reducing the repair bandwidth. In \cite{7000553}, the authors derive a lower bound on the number of packet transmissions at the MSR point for error-free broadcast partial repair, and provide an explicit code construction for a special case. The information flow graph construction in \cite{7000553} does not capture the relation between the storage capacity and the repair bandwidth, thus focusing only on one of the extreme points on the storage-repair bandwidth trade-off curve, corresponding to the MSR point. In this paper, we obtain the entire optimal trade-off curve. 

In \cite{8277979}, the authors present the storage-repair bandwidth trade-off for clustered storage networks, where multiple nodes within a cluster fail. This is close to the partial failure model that we consider since each cluster is equivalent to a node with multiple memory units, and failure of a subset of memory units in a node is the same as partial failure. However, we consider partial failures at multiple nodes, which is equivalent to failure of multiple nodes in multiple clusters, and we also consider broadcast transmissions from the non-failed nodes.
Another model studied in the literature is the centralized repair model \cite{8469091,8638804}, in which the surviving nodes transmit the repairing packets to a centralized node, which then repairs the failed nodes. The centralized model does not require the nodes being repaired to exchange data like the cooperative repair model, thus making the repair process simpler; moreover, the storage-repair bandwidth trade-off achieved with centralized repair is better than that with cooperative repair. In \cite{8469091,8638804}, it is shown that cooperative repair achieves the minimum repair bandwidth of centralized repair under exact repair but at a slightly higher storage cost. In \cite{8638804}, it is shown that the optimal functional MBR point for centralized repair of multiple nodes is not achievable under exact repair. Unlike the centralized repair model, in this paper we consider broadcast repair. The broadcast repair model is almost equivalent to the centralized multi-node repair model studied in \cite{8469091,8638804}. Therefore, the codes that we construct for broadcast repair are directly applicable to the centralized repair model as well. The main difference between the broadcast and centralized repair models is that while centralized repair employs two phases in the repair process, where a centralized entity first receives the transmissions of the surviving nodes, and then repairs the failed nodes by passing messages to them, broadcast repair employs a single phase of broadcast transmissions from surviving nodes, and therefore, is simpler and more efficient.

Since random linear coding is asymptotically optimal for network coding \cite{1705002}, it is also asymptotically optimal for functional repair in distributed storage \cite{5550492}. However, random linear coding is not an attractive scheme for distributed storage in practice due to large overhead and high computational complexity of the Gaussian elimination-based decoding  \cite{shrader2007packet}. Therefore, most literature on regenerating codes focus on exact repair due to limited overhead and predictability of the system when the contents of the storage nodes do not change with time. However, it is shown in \cite{6062413} that a large portion of the functional repair trade-off curve cannot be achieved under exact repair. It is further shown in \cite{6804941} that there is a non-vanishing gap between the exact repair trade-off curve and the functional repair trade-off curve. Hence, functional repair has been studied in \cite{DBLP:journals/corr/abs-1809-08138,6620243,6283043} in an attempt to construct regenerating codes achieving a trade-off closer to the optimal while incurring low overhead. Hollmann and Poh in \cite{6620243} viewed a regenerating code as a collection of sets of subspaces of a vector space, in which reconstruction corresponds to generating the vector space while repair corresponds to generating a subspace. This differs from the purely coding theoretic view where the file is encoded using a generator matrix, like in \cite{5961826}. With respect to the vector space interpretation of linear functional repair codes characterized in \cite{6620243}, it is stated in \cite{DBLP:journals/corr/abs-1809-08138} - ``The properties of the code are determined entirely by the manner in which the various spaces intersect. The advantage of the geometric perspective is that many classical geometric objects have nice, well-understood properties in terms of how spaces embedded in these objects intersect''. In \cite{DBLP:journals/corr/abs-1809-08138}, several constructions for \textit{strictly functional repair} are proposed using a projective geometry viewpoint for a limited set of parameters, that tolerate multiple node failures. The constructions in \cite{DBLP:journals/corr/abs-1809-08138} use well-known combinatorial objects to construct spaces while controlling how they intersect. A simple example to illustrate the point is a set of three non-concurrent lines in a plane, which gives an exact repair code with parameters $n=3, k=d=2$. In this example, each node stores two points lying on a distinct line. During repair, two nodes deliver one point each to the newcomer, on receiving which the newcomer reconstructs the line previously stored by the failed node.

Despite these works, a general construction of functional regenerating codes with low overhead, achieving all the points on the trade-off curve between storage capacity and repair bandwidth, has not been achieved. This paper addresses this problem by proposing an intuitive and practical construction to achieve functional repair for almost all admissible parameters. Our construction uses the geometric view of a linear repair code, and frames the relationship between the local intersection and global independence properties of the subspaces. In the conference publication \cite{8849496}, we presented a code construction achieving the MBR point and an interior point on the optimal storage-repair bandwidth trade-off for full node repair. In this paper, we modify and improve that scheme, and extend it to achieve all the points on the trade-off curve for full node repair as well as partial repair.

\subsection{Our contributions}
\begin{enumerate}
    \item We derive the optimal storage-repair bandwidth trade-off for broadcast repair of multiple partial failures, and discuss the advantage of repairing multiple failed nodes simultaneously while leveraging unerased data remaining in the failed nodes for repair.
    \item We derive invariant conditions that are sufficient for a functional repair code to be optimal and feasible. These conditions render structure to the functional repair problem, and provide insights into how to construct a feasible functional repair code.
    \item We present explicit codes that achieve the optimal cut-set bound for functional repair of multiple nodes with high probability, for a large number of feasible parameters, in particular, for values of $k$ which are divisible by $r$.
    \item The proposed scheme employs a subpacketization level that scales linearly in the code parameters, has low computational complexity during the repair process, and introduces low overhead unlike random linear coding.
\end{enumerate}

The rest of the paper is organized as follows. The system model is introduced in Section \ref{sec:nm}. The storage-bandwidth trade-off for partial repair is derived in Section \ref{sec:tradeoff}. Code constructions for multiple node repair and partial repair are presented in Section \ref{sec:multiple}. Results and discussions are presented in Section \ref{sec:rd}. We conclude the paper in Section \ref{sec:conc}.

\section{System Model} \label{sec:nm}
Consider a wireless caching system where $n$ nodes, each with storage capacity $\alpha$ bits, store a file of size $M$ bits. We index these storage nodes by $\{ 1, \ldots, n \}$. The nodes are fully connected by a wireless broadcast medium and use orthogonal channels for data transmission. The file is divided into $P$ data packets, and the number of data packets stored in each node, which is referred to as the \textit{subpacketization} of an array code, is defined as $S = \frac{\alpha P}{M}$. The nodes store the file such that a user, modeled as a \textit{data collector (DC)}, can reconstruct the file by obtaining the contents of any $k$ nodes. This is called the \textit{reconstruction property}.

We consider a scenario in which a portion of the stored bits in the storage nodes is subject to being lost. We refer to these nodes as the \textit{faulty nodes}, and to the nodes that do not experience any losses as the \textit{complete nodes}. We assume that the repair occurs in rounds, where a repair round gets initiated when $r$ nodes experience partial failures of $\alpha - \alpha_{1}$ bits, where \textit{the number of non-corrupted bits} $\alpha_1 \triangleq \rho \alpha , \rho \in [0,1]$, is a $\rho^{th}$ fraction of $\alpha$. Thus, a single repair round repairs $r$ faulty nodes, where $r\leq n-d$. There is no loss during a repair round, during which the lost bits in the faulty nodes are repaired with the help of bits transmitted from $d$ complete nodes, $k \leq d \leq n-r$, called the \textit{helper nodes}, and the remaining bits that have not been lost in each of the faulty nodes. In general, the repair is functional, i.e., the repaired portion may not be the same as the original portion, but the repaired nodes satisfy the reconstruction property. See Table \ref{notation} for a list of the parameters associated with a regenerating code.

\begin{table}[t] 
		\caption{Notation.}
		\label{notation}
		\begin{center}
			\begin{tabular}{|l|p{75mm}|}
				\hline
				$n$ & Number of storage nodes\\
				$k$ & Minimum number of nodes required for file reconstruction\\
                $r$ & Number of repaired nodes (newcomers) in each repair round\\
                $d$ & Number of helper nodes \\
				$M$ & File size in number of bits\\
                $\alpha$ & Number of stored bits per node\\
                $\alpha_1$ & Number of non-corrupted bits in a faulty node\\
                $\rho$ & Fraction of non-corrupted bits in a faulty node, i.e., $\alpha_1=\rho \alpha$\\
                $\beta$ & Number of transmitted bits per helper node\\
                $\gamma$ & Total repair bandwidth, i.e., $\gamma = d\beta$\\
				\hline
			\end{tabular}
		\end{center}
	\end{table}

\subsection{Information flow graph}

The repair dynamics of the network can be represented by an information flow graph that evolves in time. See Fig. \ref{fig:flow1} for an illustration. It is a directed acyclic graph consisting of seven types of nodes: a single source node $S$ (orange), storage nodes $x_{in}^i$ (blue), $x_{mid}^i$ (gray), and $x_{out}^i$ (green), failed portion of the nodes $x_f^i$ (red), auxiliary nodes $h_i$ (yellow), and a DC node denoted by $DC$ (cyan). 
Initially, each complete storage node, denoted by $x^i, i =1,\ldots, n$, is represented by two vertices: an input vertex $x_{in}^i$ and an output vertex $x_{out}^i$, which are connected by a directed edge $x_{in}^i \rightarrow x_{out}^i$ with capacity $\alpha$.
A faulty node is represented by four vertices: an input vertex $x_{in}^i$, an intermediate vertex $x^{i}_{mid}$ that is connected to $x_{in}^i$ by a directed edge $x_{in}^i \rightarrow x_{mid}^i$ of capacity $\alpha$, an output vertex $x_{out}^i$ that is connected to $x_{mid}^i$ by a directed edge $x_{mid}^i \rightarrow x_{out}^i$ of capacity $\alpha_1$, and a failed vertex $x_f^i$ (red) that is connected to $x_{mid}^i$ by a directed edge $x_{mid}^i \rightarrow x_{f}^i$ of capacity $\alpha-\alpha_1$. The failed vertex represents the corrupted portion of data in a storage node. 

Each vertex in the graph at any given time has two modes, active or inactive, depending on its availability. At the beginning, the source node $S$ is active and it transmits data to $n$ storage nodes such that a $DC$ can retrieve the file from any $k$ nodes. This is modeled by adding an edge from $S$ to all the input vertices of the storage nodes, $S \rightarrow x_{in}^i, i \in [n]$, with capacity $\infty$\footnote{Note that adding an edge with capacity $\infty$ means that all the information in the node sending the data is available in the input vertices of the nodes receiving the data.}. From this point onwards, the source node becomes inactive, and the storage nodes become active. The directed edge of capacity $\alpha$ between the input vertex and the output vertex representing each storage node allows only $\alpha$ bits of information to propagate forward through the storage node, therefore modeling the storage capacity of the node.

When $r$ nodes experience partial failure of $\alpha - \alpha_1$ bits each, in the $s$-th round, the repair process is triggered and $r$ newcomers join the system. Note that a newcomer represents the corresponding node being repaired. A newcomer $x^i$, where $i = sn+j, j\in [n]$, represents the node $x^j$ after the $s$-th round. For example, if $n=4$ as in Fig. 1, the storage nodes in the beginning are $x^1, \ldots, x^4$. Consider that nodes $x^1$ and $x^2$ fail. After one round of repair, newcomer nodes $x^5$ and $x^6$ represent the repaired nodes $x^1$ and $x^2$ respectively, and $x^7$ and $x^8$ are copies of the nodes $x^3$ and $x^4$, respectively. The lost data is regenerated at the newcomers by receiving functions of the stored data from $d$ helper nodes.
The $d$ helper nodes are connected to the corresponding auxiliary nodes, denoted by $h^i$, with a directed edge $x_{out}^i \rightarrow h^i$ of capacity $\beta$, which denotes the number of bits broadcasted by $x^i$. Each auxiliary node $h^i$ is connected with infinite capacity links to all the newcomers. This represents the broadcast nature of the transmission medium.
\begin{definition} 
The repair bandwidth $\gamma=d\beta$ is defined as the total number of bits the helper nodes broadcast in a repair round.
\end{definition}
We model a newcomer with two vertices $x_{in}^i$ and $x_{out}^i$ and a directed edge $x_{in}^i \rightarrow x_{out}^i$ with capacity $\alpha$. The newcomer $x^i, i={sn+(s-1)r+j}, j\in [r]$, uses the $\alpha_1$ bits from the corresponding node being repaired. This is captured in the flow graph by edges with capacity $\alpha_1, x_{mid}^i \rightarrow x_{out}^i, i=(s-1)n+(s-1)r+j, j\in [r]$, followed by the edges with infinite capacity between the output vertices of the node being repaired and the newcomers.

A DC corresponds to a request to reconstruct the file. DCs connect to any subset of $k$ active nodes and retrieve all the stored data in these nodes, represented with edges with infinite capacity from the active nodes to a node $DC$.

A cut-set in the information flow graph is a subset of edges such that there is no path from the source node $S$ to the $DC$ that does not go through any of the edges in the cut-set. A cut partitions the graph into two disjoint sets of vertices, denoted by the pair $(U,\bar{U})$, where $U$ is the set of vertices on the left of the cut, and $\bar{U}$ is the remaining vertices on the right of the cut, assuming that the direction of all edges in the graph is from left to right. We define the capacity of a cut-set as the sum of its edge capacities, and the \textit{min-cut} of a graph as the cut-set with the minimum capacity among all the cut-sets.
\begin{proposition} \cite{5550492} \label{prop:1}
Consider any given finite information flow graph $\cal G$, with a finite set of DCs. If the min-cut separating the source from each DC is larger than or equal to the file size $M$, then there exists a
linear network code defined over a sufficiently large finite field $\mathbb F$
(whose size depends on the graph size) such that all DCs can reconstruct the original file. Further, randomized network coding guarantees that all collectors can reconstruct the file with probability that can be driven arbitrarily close to $1$ by increasing the field size.
\end{proposition}

Following Proposition \ref{prop:1}, for the information flow graph construction described above, we find the minimum cut over all possible failure combinations. We enumerate cuts, denoted by $\chi$, as $\chi_1, \chi_2, \chi_3$ (see Fig. \ref{fig:flow1}). In Section \ref{sec:tradeoff}, we demonstrate how to find the min-cut for a specific example, and finally in the proof for Theorem \ref{theorem 1} in Section \ref{sec:tradeoff}, we describe the process for finding the min-cut for a general information flow graph.

\section{Storage-Bandwidth Trade-off for Partial Repair} \label{sec:tradeoff}

Consider the scenario illustrated in Fig. \ref{fig:flow1}, where $n=4, k=2,d=2$ and $r=2$. The capacity of cut $\chi_1$ is $2\alpha_1 + 2\beta$, while the capacity of cut $\chi_2$ is $2\alpha$. Then the min-cut is $\min\{2\alpha_1 + 2\beta, 2\alpha\}$. From Proposition \ref{prop:1}, to ensure that the file can be reconstructed by the $DC$, $\min\{2\alpha_1 + 2\beta, 2\alpha\} \geq M$.

\begin{figure*}
		\centering
		\includegraphics[scale=0.9]{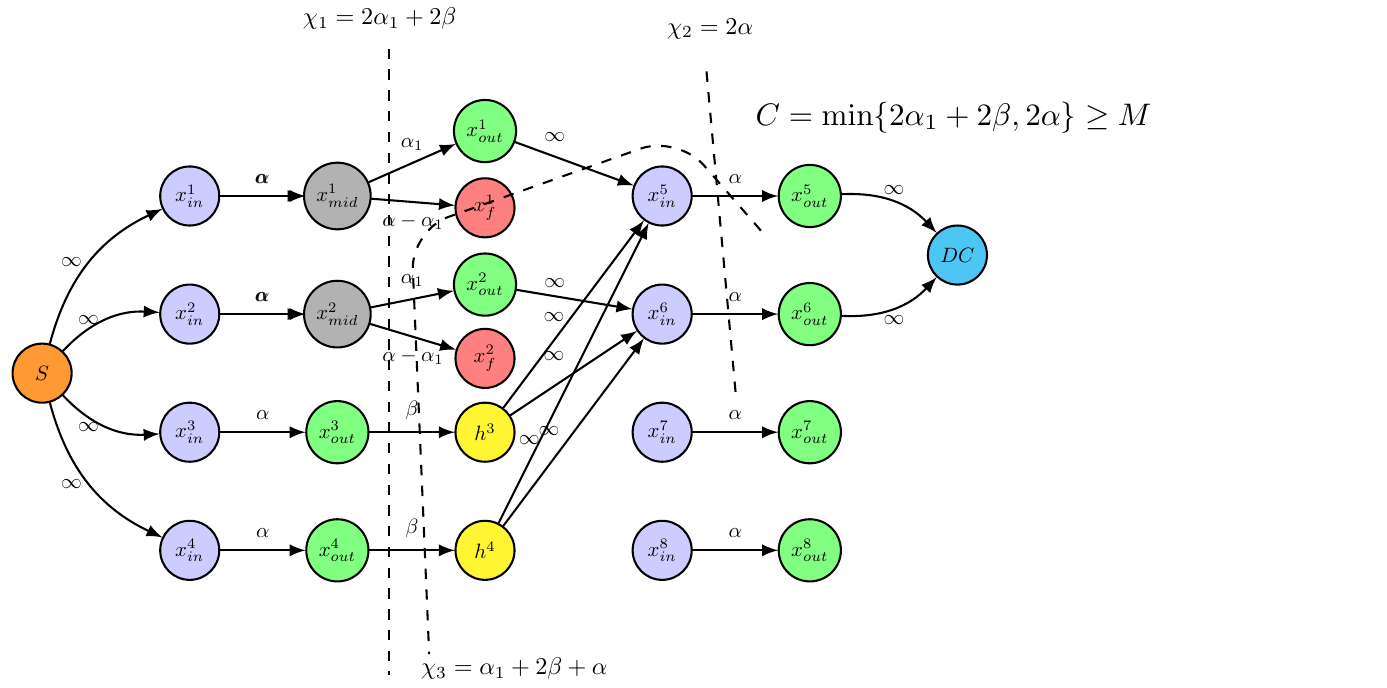} 
        \caption{Information flow graph $\cal G$ with $n=4, k=r=2$, one repair round and cuts $\chi_1,\chi_2, \chi_3$.}
		\label{fig:flow1}
\end{figure*}

For each set of parameters $(n, k,d, \gamma, \alpha, r, \rho)$, there is a family of information flow graphs, each of which corresponds to a particular evolution of node failures/repairs. We denote this family of directed acyclic graphs by $\cal G$$(n, k,d, \gamma, \alpha, r, \rho)$. An $(n, k,d, \gamma, \alpha, r, \rho)$ tuple is feasible, if a code with storage $\alpha$ and repair bandwidth $\gamma$ exists.

\begin{theorem}\label{theorem 1}
For any $\alpha \geq \alpha^{*}(n,k,d,\gamma, r, \rho)$, the points $(n,k,d, \gamma, \alpha, r, \rho)$ are feasible, and linear network codes suffice to achieve them. It is information theoretically impossible to achieve points with $\alpha < \alpha^{*}(n,k,d,\gamma, r, \rho)$. If $r$ divides $k$, the threshold function $\alpha^{*}(n,k,d,\gamma, r, \rho)$ is given by:
\begin{align}
\alpha^{*}(n,k,d,\gamma, r, \rho)=\left\{
                \begin{array}{ll}
                  \frac{M}{k} \hspace{12.5mm} \gamma \in \left[ f(0), \infty \right) \\
                  \frac{M - g(i)\gamma}{k-ir(1-\rho)} \hspace{2mm} \gamma \in \left[ f(i), f(i-1)\right] 
                \end{array}
              \right.
\end{align}
where, for $i={1,2,\ldots , \frac{k}{r}-1},$
\begin{align}
f(i)&\triangleq\frac{2Md(1-\rho)}{(2k-ir(1-\rho))(i+1)+\frac{2k}{r}(d-k)},\\
g(i)&\triangleq \left(2d-2k+r+ir\right)\frac{ir}{2d}.
\end{align}
\end{theorem}
\begin{proof}
Proof in Appendix \ref{appendix1}.
\end{proof}
\begin{corollary} \label{corollary1}
The minimum storage point is achieved by the pair
\begin{align}
(\alpha^*_{MSR}, \gamma^*_{MSR})=\left(\frac{M}{k}, \frac{Mrd(1-\rho)}{k(d-k+r)}\right).
\end{align}
\end{corollary}

\begin{corollary} \label{corollary2}
The minimum repair bandwidth point is achieved by the pair
\begin{align}
(\alpha^*_{MBR}, \gamma_{MBR}^*)&=\left( \frac{2Md}{k(2d-(k-r)(1-\rho))}, \right.\\
& \hspace{0.9cm} \left. \frac{2Mrd(1-\rho)}{k(2d-(k-r)(1-\rho))}\right).
\end{align}
\end{corollary}

MSR and MBR codes attain the points in Corollary \ref{corollary1} and Corollary \ref{corollary2}, respectively. 

\begin{remark}
For $\rho=0$ and $r=1$, i.e., complete failure of exactly one node, the model is equivalent to that in \cite{5550492}, and the trade-off curve from Theorem \ref{theorem 1} coincides with the trade-off curve in \cite{5550492}. Similarly, for $\rho=0 \text{ and } r>1$, i.e., multiple complete failures, the trade-off curve from Theorem \ref{theorem 1} coincides with the trade-off curve in \cite{7459908}. 
\end{remark}

Theorem \ref{theorem 1} provides a piecewise linear trade-off curve that defines the optimal storage capacity as a function of the repair bandwidth when $k$ is divisible by $r$, as shown in Fig. \ref{fig:tradeoff}. The curve is linear between points with $\gamma = f(i) $ and $\gamma = f(i-1)$, $ i=1,\ldots, \frac{k}{r}-1$, where $f(i)$ is a decreasing function of $i$ and defines the position of the corner points of the piecewise linear curve. All the points lying above the curve defined by Theorem \ref{theorem 1} are achievable. Corollary \ref{corollary1} defines the MSR point, that is, the point on the trade-off curve that has the lowest feasible storage capacity, while Corollary \ref{corollary2} defines the MBR point, that is, the point on the trade-off curve that has the lowest feasible repair bandwidth (see Fig. \ref{fig:tradeoff}).

\begin{theorem}\label{theorem 2}
In the same context as in Theorem \ref{theorem 1}, if $r$ does not divide $k$, let $p\triangleq \lfloor \sfrac{k}{r} \rfloor$ such that $k_0\triangleq pr$. Find $t^*\in [0:p-2]$ such that $\frac{d-k_0+t^*r}{r}\leq \frac{d-k_0}{k-k_0}\leq \frac{d-k_0+(t^*+1)r}{r}$. Also define $k^\prime \triangleq k\rho + (1-\rho)k_0$. Then the threshold function $\alpha^{*}(n,k,d,\gamma, r, \rho)$ is given by:
\begin{align}
\alpha^{*}=\left\{
                \begin{array}{ll}
                  \frac{M - g(i)\gamma}{k-ir(1-\rho)} \hspace{20.2mm} \gamma \in [f(i), f(i-1)] ,\\
                 \hspace{34mm}0\leq i \leq t^*-1 \\ \\
                  \frac{M - g(t^*)\gamma}{k-t^*r(1-\rho)} \hspace{19.5mm} \gamma \in \left[ f^\prime, f(t^*-1)\right] \\ \\
              \frac{M - [g(t^*)+\frac{d-k_0}{d}]\gamma}{k^\prime -t^*r(1-\rho)} \hspace{6.1mm} \gamma \in \left[ f(t^*), f^\prime \right] \\ \\
              \frac{M - [g(i)+\frac{d-k_0}{d}]\gamma}{k^\prime -ir(1-\rho)} \hspace{7mm} i\geq t^*+1, \\
           \hspace{34.8mm}   \gamma \in \left[ f(i), f(i-1) \right]
                \end{array}
              \right.
\end{align}
where $i={0,1,\ldots , \frac{k}{r}-1}$, and $f, g \text{ and } f^\prime$ are defined as follows:
\begin{align}
f(i)&\triangleq \left\{
\begin{array}{ll}
\infty \hspace{44.7mm} i=-1\\ \\
\frac{2Md(1-\rho)}{(2k-r(i+1)(1-\rho))i+\frac{2k}{r}(d-k_0)} \hspace{10.3mm} i\leq t^*-1 \\ \\
\frac{2Md(1-\rho)}{(2k^\prime-r(i+1)(1-\rho))i+\frac{2k^\prime (d-k_0)}{r}+ d-k_0}  \hspace{8.3mm} i\geq t^*
\end{array}
\right.
\end{align}
\begin{align}
g(i)&\triangleq \left(2d-2k_0+r+ir\right)\frac{ir}{2d} \\ 
f^\prime &\triangleq 
\frac{2Md}{\left[\frac{2(d-k_0)(k-k_0-r)}{k-k_0}+(t^*+1)r\right]t^* + \frac{2k(d-k_0)}{(k-k_0)(1-\rho)}}.
\end{align}
\end{theorem}
\begin{proof}
Proof in Appendix \ref{appendix2}
\end{proof}

Theorem \ref{theorem 2} provides the trade-off curve when $k$ is not divisible by $r$. In this case, there is an additional corner point where $\gamma = f'$ on the piecewise linear trade-off curve where the slope changes, that depends on the capacity contribution of $k-k_0 < r$ nodes. The position of the additional corner point depends on the value of $t^*$ that satisfies the condition $\frac{d-k_0+t^*r}{r}\leq \frac{d-k_0}{k-k_0}\leq \frac{d-k_0+(t^*+1)r}{r}$. 



\section{Code Construction} \label{sec:multiple}

In this section, we present a framework for constructing explicit storage and repair schemes that can achieve storage-repair bandwidth pairs that are on the optimal trade-off curve. First, we provide a few preliminary concepts that are vital for the construction. 

\subsection{Subspace view}
 
Consider that a node stores $S$ linearly independent data packets $y_1, \ldots, y_{S}$ consisting of symbols in the finite field $ \mathbb{F}_{q^l}$. For simplicity, assume that each data packet consists of exactly one symbol in $\mathbb{F}_{q^l}$. Finite field symbols in $\mathbb{F}_{q^l}$ can be viewed as $l-$dimensional vectors over $\mathbb{F}_{q}$. 
Linear operations performed on the stored symbols correspond to linear operations on their vector representations in $\mathbb{F}_{q}$. Hence, we say that the node stores a subspace of dimension $S$, denoted by $W_i = \spn\{y_i\}, i=1,\ldots S$. For a set of nodes denoted by $\mathcal{A}$, the subspace stored by $\mathcal{A}$ is denoted by $W_{\mathcal{A}} = \sum_{i\in \mathcal{A}} W_i$. The sum of two vector spaces $W_1$ and $W_2$ is defined as $W_1 + W_2 = \{ w_1 + w_2: w_1 \in W_1, w_2 \in W_2 \}$. Note that the sum of two vector spaces is not in general equal to their union. We use the notation $\dim(\cdot)$ for the dimension of a vector space. The subspace view of linear storage codes has also been used in previous works like \cite{8638804, 6620243,5962548}.

\subsubsection{Vector space dimension as an information measure}

Consider a sample set of $p$ linearly independent vectors over $\mathbb{F}_q^l$ defined as $\Omega \triangleq \{ w_1,\ldots, w_p \}$, and a function $f: \Omega \rightarrow \mathbb{F}_{q}^{l}$ that generates a random linear combination of a subset of vectors in $\Omega$. Consider a vector space $W$ over $\mathbb{F}_q^l$ generated by the vectors $\{w_1, \ldots, w_p \}$, and a collection $\Sigma$ of subspaces of $W$ that includes $W$, is closed under complement, and is closed under \textit{countable sums} of subspaces. Then the $\sigma$-algebra generated by the function $f$ on $\Sigma$ is given by:
\begin{align}
    \sigma(f) = \left\{ f^{-1}(V): V \in \Sigma \right\},
\end{align}
where $f^{-1}(V)$ provides the smallest pre-image of the subspace $V \in \Sigma$ under $f$. The dimension of a vector space $V\in \Sigma$, defined as the function $\dim : \Sigma \rightarrow \mathbb{N}$, is a measurable function on the space $(\Omega,\sigma(f))$ such that:
\begin{align}
    \dim(V) = \vert f^{-1}(V) \vert, 
\end{align}
where the notation $\vert \cdot \vert$ denotes the cardinality of a set. By considering the dimension of a vector space as an information measure on it, we can, as described in \cite{79902}, formulate identities for the dimension of vector spaces that are similar to those for Shannon information measures. We list a few identities in the next section associated with the dimension of vector subspaces.

\subsection{Identities associated with the dimension of vector spaces}

\begin{enumerate}
\item \textit{Conditional intersection:}
 We define the ``conditional intersection'' of a set of vector spaces $W_i, i=1,\ldots, t$, conditioned on a vector space $W_0$, as the largest subspace in the intersection of the vector spaces $W_i + W_0, i=1,\ldots, t$, after excluding the non-zero vectors belonging to the vector space $W_0$. Thus, we write the dimension of the conditional intersection as follows:
    \begin{align}
   \hspace{-0.3cm} \dim\left(\bigcap_{i=1}^{t} W_i \Big\vert W_0\right) \triangleq \dim\left(\bigcap_{i=1}^{t}(W_i+W_0)\setminus W_0 \right),
\end{align}
where $W_i\setminus W_j$, for two vector spaces $W_i$ and $W_j$, denotes the largest subspace of $W_i$ remaining after removing the non-zero elements belonging to $W_i \cap W_j$ from $W_i$.
For $t=2$, we note that the above identity becomes
\begin{align}
   \nonumber &\dim\left(W_1 \cap W_2 \Big\vert W_0\right)\\
    & \hspace{0.6cm} = \dim\left((W_1+W_0) \cap (W_2 + W_0)\setminus W_0\right)\\
    &\hspace{0.6cm} =\dim\left(W_1 \cap (W_2 + W_0)\setminus W_0\right), \label{reduction}
\end{align}
where Eq. \eqref{reduction} follows due to the following reasoning: consider a vector $w$ belonging to the vector space $(W_1+W_0) \cap (W_2 + W_0)$ given by $w=w_1 + w^{(1)}_0 = w_2 + w^{(2)}_0 $, where $w_1 \in W_1, w_2 \in W_2,$ and $w_0^{(1)}, w_0^{(2)} \in W_0$. Then, we also have $w' \in W_1 \cap (W_2 + W_0)$ given by $w'=w_1=w_2 + (w^{(2)}_0 -  w^{(1)}_0)$, where $w^{(2)}_0 - w^{(1)}_0 \in W_0$. Hence, the dimensions of the vector spaces $(W_1+W_0) \cap (W_2 + W_0)$ and $W_1 \cap (W_2 + W_0)$ are equal. We can also deduce the following \textit{chain rule} from Eq. \eqref{reduction}:
\begin{align}
    \nonumber &\dim\Big( W_1 \cap \left( W_2 + W_3 \right) \Big) = \dim\Big( W_1 \cap W_2 \Big) + \\
    &\hspace{0.5cm}  \dim\Big( W_1 \cap \left(W_3 + W_2 \right) \setminus W_2\Big)\\
    &=\dim\Big( W_1 \cap W_2 \Big) + \dim\Big( W_1 \cap W_3 \Big\vert W_2\Big)
\end{align}

\item $\dim\Big( W_1 + W_2 \Big) = \dim\Big( W_1 \Big)+ \dim\Big( W_2 \Big) - \dim\Big( W_1 \cap W_2 \Big)$.

\item $\dim\Big( W_1 \cap W_2 \Big) = \dim\Big( W_1 \Big) - \dim\Big( W_1 \setminus W_2 \Big)$.

\end{enumerate}

\subsection{Linearized polynomials}
An important component in our construction is the linearized polynomial and its special properties. 
A linearized polynomial 
\begin{align}
    f(x)=\sum_{i=1}^{P}a_i x^{q^{i-1}} ,\ \ \  a_i \in \mathbb{F}_{q^l},
\end{align}
can be uniquely identified from evaluations at any $P$ points $x=\theta_i \in \mathbb{F}_{q^l}, i=1,2,\ldots, P$, that are linearly independent over $\mathbb{F}_q$. The polynomial interpolation problem (that is, to determine the coefficients of $f(x)$ from the evaluations) can be written as
\begin{align}\label{moore}
    \mathbf{Q} \mathbf{a}=\mathbf{y},
\end{align}
where $\mathbf{Q}$ is the Moore matrix corresponding to the evaluation points (\cite{goss2012basic}, Chapter 1.3), $\mathbf{a}=(a_1 , \ldots, a_P)^T$, and $\mathbf{y}=(f(\theta_1),\ldots,f(\theta_P))^T$. For linearly independent evaluation points $\theta_i, i=1,\ldots, P$, $\mathbf{Q}$ is invertible, thus proving the existence of a unique solution for Eq. \eqref{moore}. 

Another relevant property of linearized polynomials is that they satisfy:
\begin{align}
    f(ax+by)=af(x)+bf(y), \hspace{0.5cm} a,b\in \mathbb{F}_q, \hspace{0.3cm} x,y \in \mathbb{F}_{q^l}.
\end{align}
In other words, given a set of points on a linearized polynomial, any linear combination over $\mathbb{F}_q$ of the points also lies on the polynomial.

\subsection{General code construction for any point on the trade-off curve with $\rho=0$ (full-node repair)}\label{general_scheme}
Substituting $i=\frac{k}{r}-\bar{j}, \bar{j}\in [\frac{k}{r}]$ in Theorem \ref{theorem 1}, we obtain the general expression for any point on the optimal storage-repair bandwidth trade-off as 
\begin{align}
(\alpha^*,\gamma^*)=\frac{M}{P^*} \Big(d-(\bar{j}-1)r,rd\Big) , \hspace{1cm} \bar{j} \in \left[\frac{k}{r}\right],
\end{align}
where $P^*=\sfrac{k}{2}\Big(2\left(d-(\bar{j}-1)r\right)-\left(k-r\right)\Big)+r\Big((\bar{j}-1)k-\frac{\bar{j}(\bar{j}-1)}{2}r\Big)$. By considering $\frac{M}{P^{*}}$ as the size in bits of one data packet stored in a node, an optimal scheme stores $d-(\bar{j}-1)r$ data packets in a node, and has $d$ helper nodes broadcasting $rd$ data packets for the repair of $r$ nodes in a repair round. 
Conversely, a scheme that stores $d-(\bar{j}-1)r$ data packets in a node, and has $d$ helper nodes broadcasting $rd$ data packets for the repair of $r$ nodes in a repair round, is optimal if the size of each data packet is $\frac{M}{P^{*}}$. Note that the points on the trade-off curve are parametrized by $\bar{j}$, and are obtained by varying $\bar{j}$.

In the following, we first state three conditions for optimal functional repair. Then, we prove the sufficiency of these conditions. Finally, we construct a general scheme that satisfies these conditions with high probability, and therefore, can achieve functional repair for any point on the trade-off curve.
\subsubsection{Conditions for an optimal scheme}\label{conditions}
The following conditions \textbf{L1},\textbf{L2}, and \textbf{L3}, are sufficient for optimal functional repair, and are described as follows:
\begin{description}
    \item[\textbf{L1:}]  For any set of nodes $\mathcal{A}$ such that $\lvert \mathcal{A} \rvert \leq \bar{j}r$, the following holds :
    $\dim\Big(\sum_{i\in \mathcal{A}} W_i\Big)=\sum_{i\in \mathcal{A}} \dim\Big(W_i\Big)$. 
    This further implies that $\dim\Big(W_{\mathcal{A}_1} \cap W_{\mathcal{A}_2}\Big)=0$, where $\mathcal{A}_1 $ and $\mathcal{A}_2$ are disjoint partitions of $\mathcal{A}$.
    
    \item[\textbf{L2:}] Given a node $A$, and a set of nodes denoted by $\mathcal{B}$ such that $\lvert \mathcal{B} \rvert \leq d-(\bar{j}-1)r$. Partition $\mathcal{B} $ into two disjoint non-empty sets $\mathcal{B}_1$ and $\mathcal{B}_2$. Then the following holds:
    \begin{align}
        &\dim\Big(W_{A} \cap \left(W_{\mathcal{B}_1} + W_{\mathcal{B}_2}\right)\Big)\\
        &= \dim\Big( W_A \cap W_{\mathcal{B}_1} \Big) +\dim\Big( W_A \cap W_{\mathcal{B}_2} \Big).
    \end{align}
    This is equivalent to the following condition: Let $S_{A}^{\mathcal{B}_i}, \mathcal{B}_i \subset \mathcal{B}, i =1,2$, be the subspace broadcasted by node $A$ to repair nodes in $\mathcal{B}_1$ and $\mathcal{B}_2$. Then, the following must hold:
    \begin{align}
    &\dim(S_{A}^{\mathcal{B}_1} \cap S_{A}^{\mathcal{B}_2}) = 0.\label{intersection_condition}
    \end{align}

    \item[\textbf{L3:}] Given a node $A$, and disjoint sets of $r$ nodes denoted by $\mathcal{R}_1,\ldots, \mathcal{R}_{\bar{j}}$, the following holds:
    \begin{align}
        \dim \left(W_A \cap W_{\mathcal{R}_{\bar{j}}} \Big\vert  \sum_{i=1}^{\bar{j}-1} W_{\mathcal{R}_i}\right) \leq r. 
    \end{align}
\end{description}
We now show that the conditions \textbf{L1}, \textbf{L2} and \textbf{L3} are sufficient for optimal functional repair by proving that the reconstruction property is satisfied if these conditions are met by a storage and repair scheme.
\subsubsection{Reconstruction}\label{general_reconstruction}
Suppose a DC accesses the nodes $1,\ldots,k$, denoted by $\mathcal{A}_{dc}$. For correct reconstruction, the data available at the $k$ nodes should span the vector space spanned by the $P$ packets of the file. Therefore, a necessary condition for successful reconstruction is $\dim(W_{\mathcal{A}_{dc}}) \geq P$. It is also a sufficient condition for the reconstruction of the file if the exact linear mapping between the packets available at the $k$ nodes and the $P$ packets of the file is known. In Section \ref{proposed_construction}, we show that if the file packets are encoded with the structure provided by linearized polynomials, the above condition is sufficient for successful reconstruction of the file. In this section, we derive the dimension of the subspace stored by $k$ nodes, assuming that the properties $\mathbf{L1}, \mathbf{L2}$ and $\mathbf{L3}$ are satisfied by the storage nodes. First we propose the following lemma.

\begin{lemma}\label{lemma:1}
Assume that $\mathbf{L1}, \mathbf{L2}$ and $\mathbf{L3}$ are satisfied. Given a node $A$, and a set of $v\leq d$ nodes denoted by $\mathcal{B}$, partition $\mathcal{B}$ into sets of $r$ nodes denoted by $\mathcal{R}_1, \ldots, \mathcal{R}_{\lfloor \sfrac{v}{r} \rfloor} $, and denote the remaining set of nodes by $\mathcal{R}'$. Then,
\begin{align}
    \dim\Big(W_A \cap W_{\mathcal{B}}\Big) = \sum_{s\geq \bar{j}}^{\lfloor\sfrac{v}{r} \rfloor} \dim \left(W_A \cap W_{\mathcal{R}_s} \Big\vert \sum_{t=1}^{\bar{j}-1} W_{\mathcal{R}_t}\right).
\end{align}
\end{lemma}
\begin{proof}
Proof in Appendix \ref{appendix3}.
\end{proof}

\begin{theorem}\label{reconstruction_thm}
If a DC accesses $k$ nodes, denoted by $\mathcal{A}_{dc}$, $k \leq d$, then, assuming that $\mathbf{L1}, \mathbf{L2}$ and $\mathbf{L3}$ are satisfied, we have
\begin{align}
    \dim(W_{\mathcal{A}_{dc}}) \geq P^{*},
\end{align}
thus satisfying the reconstruction property with optimal storage and repair bandwidth.
\end{theorem}
\begin{proof}
We have
\begin{align}
    &\dim\Big(\sum_{i=1}^{k}W_{i}\Big) \\
    &= \sum_{i=1}^{k} \left[ \dim\Big(W_i\Big)  - \dim\Big(W_i \cap \sum_{u=1}^{i-1}W_{u}\Big)\right] \\
    &= \sum_{i=1}^{k} \dim\Big(W_i\Big) - \\
    & \hspace{1cm} \sum_{i=1}^{k} \sum_{s=\bar{j}}^{\lfloor \sfrac{(i-1)}{r}\rfloor} \dim\Big(W_i \cap W_{\mathcal{R}_s}\Big\vert \sum_{t=1}^{\bar{j}-1}W_{\mathcal{R}_t}\Big) \label{stepf}
\end{align}
\begin{align}  
    &\overset{}{\geq} k(d-(\bar{j}-1)r)- \left( r(r) + \cdots + r(k-\bar{j}r) \right)\\
    &= \frac{k}{2}\left(2(d-(\bar{j}-1)r)-(k-r)\right) + \\
    & \hspace{1cm} r\left((\bar{j}-1)k-\frac{\bar{j}(\bar{j}-1)}{2}r\right)\\
    &= P^{*}
\end{align}
where Eq. \eqref{stepf} follows from Lemma \ref{lemma:1}. \end{proof}
Therefore, a storage and repair scheme that divides the file into $P$ data packets, and satisfies $\mathbf{L1}, \mathbf{L2}$ and $\mathbf{L3}$, achieves the optimal tradeoff between storage and repair bandwidth, by setting $P=P^{*}$. In the following section, we propose a scheme that satisfies $\mathbf{L1}, \mathbf{L2}$ and $\mathbf{L3}$ with high probability.
\subsubsection{Proposed code construction} \label{proposed_construction}
A file of size $M$ bits is divided into 
$P$ data packets denoted by $m_1,\ldots, m_P$. For convenience and without loss of generality, we assume that each packet consists of exactly one symbol in $\mathbb{F}_{q^l}$. Define the linearized polynomial 
\begin{align}\label{linearized_poly_code}
 f(x)=\sum_{i=1}^{P}m_i x^{q^{i-1}} , \hspace{0.6cm} m_i\in \mathbb{F}_{q^l}, 
\end{align}
in a finite field $\mathbb{F}_{q^{l}}, l \geq P$. If a DC receives evaluations of the polynomial $f(x)$ on any $P$ points in $\mathbb{F}_{q^l}$ that are linearly independent over $\mathbb{F}_q$, it can reconstruct $f(x)$ by interpolation, and thus reconstruct the file. For the rest of the section, we shall refer to evaluations of $f(x)$ on a set of linearly independent evaluation points as \textit{linearly independent evaluations}. We propose a general scheme parameterized by $\bar{j}$ that achieves the points on the optimal storage-repair bandwidth trade-off with a high probability. Each node stores $d-(\bar{j}-1)r$ linearly independent evaluations of $f(x)$.

We set the size of the finite field to $q^l$, where $l\geq (n-r)(d-(\bar{j}-1)r)$, and store $d-(\bar{j}-1)r$ linearly independent evaluations of $f(x)$ on the nodes $1,\ldots, n-r$. Subsequently, the contents of the remaining $r$ nodes are generated by the nodes $1,\ldots,d$ by using the repair scheme described in the next subsection, as if the $d$ nodes are helper nodes repairing the nodes $n-r+1,\ldots, n$. This ensures that the conditions \textbf{L1}, \textbf{L2}, and \textbf{L3} are satisfied in the initial storage round.

\subsubsection{Repair scheme} \label{repair_scheme}
The following repair scheme satisfies the properties $\mathbf{L1}, \mathbf{L2}$ and $\mathbf{L3}$, after an arbitrary number of repair rounds, with high probability. During repair, the $d$ helper nodes, denoted by the set $\mathcal{H}$, transmit $r$ packets each to repair $r$ newcomers, enumerated as $\{n_1, \ldots, n_r\} = \mathcal{N}$. Node $h_i \in \mathcal{H}$, for $i=1,\ldots, d$, transmits $r$ random linear combinations of packets in set $\mathcal{A}, \vert \mathcal{A} \vert = r+e, 0\leq e \leq d-\bar{j}r$, where set $\mathcal{A}$ consists of $r+e$ packets sampled randomly from the  packets stored in node $h_i$, and $e$ is a free parameter that can be tuned to optimize the performance. The packets transmitted by the helper nodes $h_i$ in set $\mathcal{H}$, enumerated as $w_{h_i,1},\ldots, w_{h_i,r}$, are received and arranged by each newcomer in a matrix $\mathbf{Y}$ of dimensions $(d-(\bar{j}-1)r)\times \bar{j}r$ in the following manner:
\begin{align}
     \mathbf{Y}^T= \left[ \begin{array}{cccc}
        w_{h_1,1} & w_{h_2,1} & \cdots & w_{h_{d-(\bar{j}-1)r},1} \\ 
        \vdots & \ddots & \ddots & \vdots \\
         w_{h_1,r} & w_{h_2,r} & \cdots & w_{h_{d-(\bar{j}-1)r},r} \\
        w_{h_{r+1},1} & w_{h_{r+2},1} & \cdots & w_{h_{d-(\bar{j}-2)r},1} \\
        \vdots & \ddots & \ddots & \vdots \\
        w_{h_{r+1},r} & w_{h_{r+2},r} & \cdots & w_{h_{d-(\bar{j}-2)r},r} \\
        w_{h_{2r+1},1} & w_{h_{2r+2},1} & \cdots & w_{h_{d-(\bar{j}-3)r},1} \\
        \vdots & \ddots & \ddots & \vdots \\
        w_{h_{(\bar{j}-1)r+1},r} & \cdots & \cdots & w_{h_d,r}
    \end{array} \right].\label{permuted_Y}
\end{align}
Note that all the $r$ packets received from the $d$ helper nodes are present in matrix $\mathbf{Y}$ with a certain symmetrical arrangement.

Now, indexing the rows of matrix $\mathbf{Y}^T$ from $0$ to $\bar{j}r-1$, the following rotation operation is done on each row of $\mathbf{Y}^T$:
\begin{align}
     \text{ Rotate}_{(g \text{ mod }r)}\left( \text{row }g \right), \hspace{0.5cm} g=0,\ldots, \bar{j}r-1,
\end{align}
where the function $Rotate_{\sigma}(v)$ applies a circular rotation to the vector $v$ by $\sigma$ positions. Thus, newcomer $n_i, i = 1,\ldots, r$, obtains matrix $\mathbf{Y}$ such that each row contains packets from $\bar{j}r$ distinct helper nodes. 

Newcomer $n_i$ computes random linear combinations over $\mathbb{F}_q$ of the $\bar{j}r$ packets in each row of $\mathbf{Y}$, and stores the resultant packets in its memory.



In the following, we argue that the conditions \textbf{L1}, \textbf{L2} and \textbf{L3} are satisfied by the above repair scheme with high probability.
\begin{description}
    \item[\textbf{L1:}] Consider a non-zero random linear combination $v_{p+1}$ of a set $V$ of $p$ linearly independent vectors $v_1,\ldots, v_p$. Then, the set of vectors $\{v_{p+1}\}\cup V', V'\subset V, \vert V' \vert = p-1 $ forms a basis. Moreover, a set containing $t$ non-zero random linear combinations of the $p$ vectors in $V$, and any $p-t$ vectors in $V$, forms a basis with a very high probability if the field size is sufficiently large. Therefore, since each newcomer stores a random linear combination of $\bar{j}r$ packets received from $\bar{j}r$ distinct nodes, the property \textbf{L1} is satisfied with a high probability which approaches $1$ if $q$, the size of the base field, is sufficiently large. 
    
    \item[\textbf{L2:}] Helper $h_i$ transmits $r$ random linear combinations of $r+e$ linearly independent packets in each repair round to repair a group of $r$ newcomers. Consider that the helper $h_i$ repairs two disjoint sets of $r$ newcomers $\mathcal{R}_1$ and $\mathcal{R}_2$. The packets transmitted by $h_i$ can be written as $\mathbf{U}\mathbf{T} = \mathbf{U}\left[ \mathbf{T}_1\ \ \mathbf{T}_2 \right]$, where $\mathbf{U} \in \mathbb{F}_q^{l \times (d-(\bar{j}-1)r)}$ is the matrix representation in the base field $\mathbb{F}_q$ of the packets stored in $h_i$, and $\mathbf{T}_1, \mathbf{T}_2 \in \mathbb{F}_q^{(d-(\bar{j}-1)r)\times r}$, having $r+e$ non-zero rows each, capture the coefficients of the linear combinations of the $r+e$ packets transmitted for the repair of $\mathcal{R}_1$ and $\mathcal{R}_2$ respectively. The sufficient condition for Eq. \eqref{intersection_condition} to hold for $\mathcal{B} = \mathcal{R}_1 \cup \mathcal{R}_2$ is that matrix $\mathbf{T}$ should have full column rank. The probability that $\mathbf{T}$ has full column rank is given by the probability that $\mathbf{T}$ has at least $2r$ non-zero rows. Therefore, the probability that $\mathbf{T}$ has full column-rank is given by:
    \begin{align}
      \nonumber & \text{Pr}(\mathbf{T}\text{ has full column-rank}) \\
       & \hspace{0.5cm} = \frac{\sum_{i=0}^{2e} {r+e\choose i}{d-(\bar{j}-1)r - (r+e)\choose r+e-i}}{{d-(\bar{j}-1)r\choose r+e}}.
    \end{align}
    Consequently, if $r\geq 1$ and $e=0$, we have
    \begin{align}
        \text{Pr}(\mathbf{T}\text{ has full column-rank}) &= \frac{ {d-(\bar{j}-1)r - r\choose r}}{{d-(\bar{j}-1)r\choose r}}
        \end{align}
        \begin{align}
        \hspace{0.9cm} = \prod_{i=1}^{r}  \frac{d-(\bar{j}-1)r -2r + i}{d-(\bar{j}-1)r -r + i},
    \end{align}
    which is close to $1$ if $d$ is sufficiently large. As $e$ increases, the probability of $\mathbf{T}$ having a full column-rank increases.
    If $e\geq r$, matrix $\mathbf{T}$ has full column rank with probability 1. In general, the probability of matrix $\mathbf{T} \in \mathbb{F}_q^{(d-(\bar{j}-1)r)\times \vert \mathcal{B} \vert}$ having full column-rank can be made as high as desired if $e$ is set to an appropriate value depending on the parameters $d$ and $r$.
    
    \item[\textbf{L3:}] 
  
    This property holds for the described scheme because each helper node delivers $r$ linearly independent packets to the newcomers, which are then linearly combined with the packets received from $\bar{j}-1$ other helper nodes which are also assumed to be linearly independent. Therefore, given the packets from the $\bar{j}-1$ helper nodes, node $A$ and the $r$ nodes in the set $\mathcal{R}_1$ only have $r$ packets in common.
\end{description}

\subsection{MSR point with $\rho=0$}
To achieve the MSR point (see Corollary \ref{corollary1}), we set $\bar{j}=\frac{k}{r}$ in the code construction in Section \ref{general_scheme}. Thus the file is divided into $P=P^{*}=k(d-k+r)$ data packets. Each node stores $d-k+r$ linearly independent points on $f(x)$. 
$\mathbf{Y}$ is of dimensions $(d-k+r)\times k$, and thus each newcomer computes random linear combinations of the $k$ packets in the rows of $\mathbf{Y}$. 
In this manner, an MSR code is constructed with a subpacketization level of $S=d-k+r$ that scales linearly with $d,k$ and $r$.

\subsection{MBR point with $\rho=0$}
To achieve the MBR point (see Corollary \ref{corollary2}), we set $\bar{j}=1$ in the code construction in Section \ref{general_scheme}. The file is divided into $P=P^{*}=\frac{k}{2}(2d-k+r)$ data packets. Each node stores $d$ linearly independent points on $f(x)$. 
$\mathbf{Y}$ is of dimensions $d\times r$, and each newcomer computes random linear combinations of the $r$ packets in the rows of $\mathbf{Y}$. 

\subsection{Code construction for any point on the trade-off curve with $\rho>0$ (partial repair)} \label{sec:partial}

In this section, we propose the sufficient conditions for an optimal partial repair scheme, and then propose an extension of the repair scheme from the previous section to achieve optimal partial node repair performance with high probability. Substituting $i=\frac{k}{r}-\bar{j}, \bar{j}\in [\frac{k}{r}]$ in Theorem \ref{theorem 1}, we obtain the general expression for any point on the optimal storage-repair bandwidth trade-off for partial repair as 
\begin{align}
(\alpha^*,\gamma^*)=\frac{M}{P^*} \Big(d-(\bar{j}-1)r,rd(1-\rho)\Big) , \hspace{0.2cm} \bar{j} \in \left[\frac{k}{r}\right],
\end{align}
where $P^*=\frac{k}{2}\left(2(d-(\bar{j}-1)r)-(1-\rho)(k-r)\right) +  r(1-\rho)\left((\bar{j}-1)k-\frac{\bar{j}(\bar{j}-1)}{2}r\right)$. An optimal partial repair scheme divides a file into $P^{*}$ data packets, and each node stores $(d-(\bar{j}-1)r)$ coded packets. Instead, we consider that, for $\xi \in \mathbb{N}$, such that $\rho \xi \in \mathbb{N}$, the file is divided into $\xi P^{*}$ data packets, and each node stores $(d-(\bar{j}-1)r)\xi$ coded packets. The storage capacity and the repair bandwidth achieved is optimal for any arbitrary $\xi$.

In the following, we first state three conditions for optimal functional repair of partially failed nodes, and then prove their sufficiency. In Section \ref{code_partial}, we propose a general scheme that satisfies these conditions with high probability, and therefore achieve functional repair for any point on the trade-off curve.
\subsubsection{Conditions for an optimal scheme}\label{conditions_partial}
The properties $\textbf{L1}$ and $\textbf{L2}$ remain the same as in Section \ref{general_scheme}. The property $\mathbf{L3}$ is described as follows:
\begin{description}
     \item[\textbf{L3:}] Given a node $A$, and $\bar{j}$ disjoint sets of $r$ nodes denoted by $\mathcal{R}_1,\ldots, \mathcal{R}_{\bar{j}}$, the following property holds:
         \begin{align}
        \dim \left(W_A \cap W_{\mathcal{R}_{\bar{j}}} \Big\vert  \sum_{t=1}^{\bar{j}-1} W_{\mathcal{R}_t}\right) \leq (1-\rho)r\xi. 
    \end{align}
\end{description}

We now show that the above conditions are sufficient for optimal functional repair.
\subsubsection{Reconstruction}\label{reconstruction_partial}
Given that $\mathbf{L1}, \mathbf{L2}$ and $\mathbf{L3}$ are satisfied, the dimension of the space obtained by a DC accessing any $k$ nodes is given by
\begin{align}
    &\dim\Big(\sum_{i=1}^{k}W_{i}\Big) = \sum_{i=1}^{k} \left[ \dim\Big(W_i\Big) - \right.\\
    & \hspace{3.6cm} \left. \dim\Big(W_i \cap \sum_{u=1}^{i-1}W_{u}\Big)\right] \\
    &\overset{}{=} \sum_{i=1}^{k} \dim\Big(W_i\Big) - 
\end{align}
\begin{align}
    & \hspace{1cm}  \sum_{i=1}^{k} \sum_{s=\bar{j}}^{\lfloor \sfrac{(i-1)}{r}\rfloor} \dim\Big(W_i \cap W_{\mathcal{R}_s}\Big\vert \sum_{t=1}^{\bar{j}-1}W_{\mathcal{R}_t}\Big)\\
    \nonumber &\overset{}{\geq} k\xi \Big(d-(\bar{j}-1)r\Big)- \Big( (1-\rho)\xi r(r) + \\ 
   & \hspace{4cm}  \cdots + (1-\rho)\xi r(k-\bar{j}r) \Big)\\
  \nonumber  &= k\xi \Big(d-(\bar{j}-1)r\Big)- \xi (1-\rho) \Big( r(r) + \cdots + r(k-r) \Big)\\
  &+ \xi (1-\rho)r\Big( \left(k-(\bar{j}-1)r\right)+  \cdots + (k-r)\Big)\\
   &= \frac{k\xi}{2}\left(2(d-(\bar{j}-1)r)-(1-\rho)(k-r)\right)\\
   & \hspace{2cm} + \xi r(1-\rho)\left((\bar{j}-1)k-\frac{\bar{j}(\bar{j}-1)}{2}r\right),
\end{align}
which is equal to $P^{*}$, thus proving optimality.
Therefore, if we have a repair scheme for which $\mathbf{L1}, \mathbf{L2}$ and $\mathbf{L3}$ are satisfied after an arbitrary number of repair rounds, the value of $P$ can be set to $P^{*}$, thus achieving the optimal performance in terms of storage and repair bandwidth. 

\subsubsection{Proposed code construction}\label{code_partial}
The general code construction for partial repair is an extension of the one presented in Section \ref{general_scheme}. 

We consider that, for $\xi \in \mathbb{N}$, such that $\rho \xi \in \mathbb{N}$, the file is divided into $\xi P^{*}$ data packets, and each node stores $(d-(\bar{j}-1)r)\xi$ coded packets. Thus, when $(1-\rho)(d-(\bar{j}-1)r)\xi$, where $ 0\leq \rho < 1$, packets are erased on each of the $r$ faulty nodes, we have an integer number of erased packets, assuming that $\rho \xi \in \mathbb{N}$. The linearized polynomial $f(x)$ is constructed with these packets as coefficients, similarly to Eq. \eqref{linearized_poly_code}. Node $i, i=1,\ldots, n$, stores $S=\Big(d-(\bar{j}-1)r\Big)\xi$ linearly independent evaluations of $f(x)$, enumerated as $w_{i,j}, j=1,\ldots, (d-(\bar{j}-1)r)\xi$. 
\subsubsection{Repair scheme}
We assume that $(1-\rho)(d-(\bar{j}-1)r)\xi$ packets are erased on each of the $r$ faulty nodes. Consider that the indices of the $r$ faulty nodes are denoted by the set $\mathcal{N}$, and the indices of the helper nodes are denoted by $\mathcal{H}$. 
During repair, the helper node $ h \in \mathcal{H}$ transmits $(1-\rho)\xi r$ random linear combinations of a set of $(1-\rho) (r+e)\xi $ randomly sampled packets. 
The number of unerased packets on a faulty node $n_i \in \mathcal{N}, i=1,\ldots, r$, is given by $\rho \xi\left(d-(\bar{j}-1)r\right)$. A set of $ r\xi \rho $ packets are randomly sampled from the unerased packets in the faulty node $n_i$, which are added to the $(1-\rho) r\xi $ packets transmitted by a helper node, thus making $r\xi$ packets per helper node. Faulty node $n_i$ arranges these packets to form a matrix $\mathbf{Y}$ of dimensions $(d-(\bar{j}-1)r)\xi \times \bar{j}r$, such that each row contains packets from $\bar{j}r$ distinct helper nodes. The conditions \textbf{L1}, \textbf{L2} and $\textbf{L3}$ are satisfied with high probability, and the rest of the repair scheme proceeds similar to that in Section \ref{repair_scheme}.

\section{Results and Discussion} \label{sec:rd}

\begin{figure}[t]
    \centering
    \includegraphics[scale=0.87]{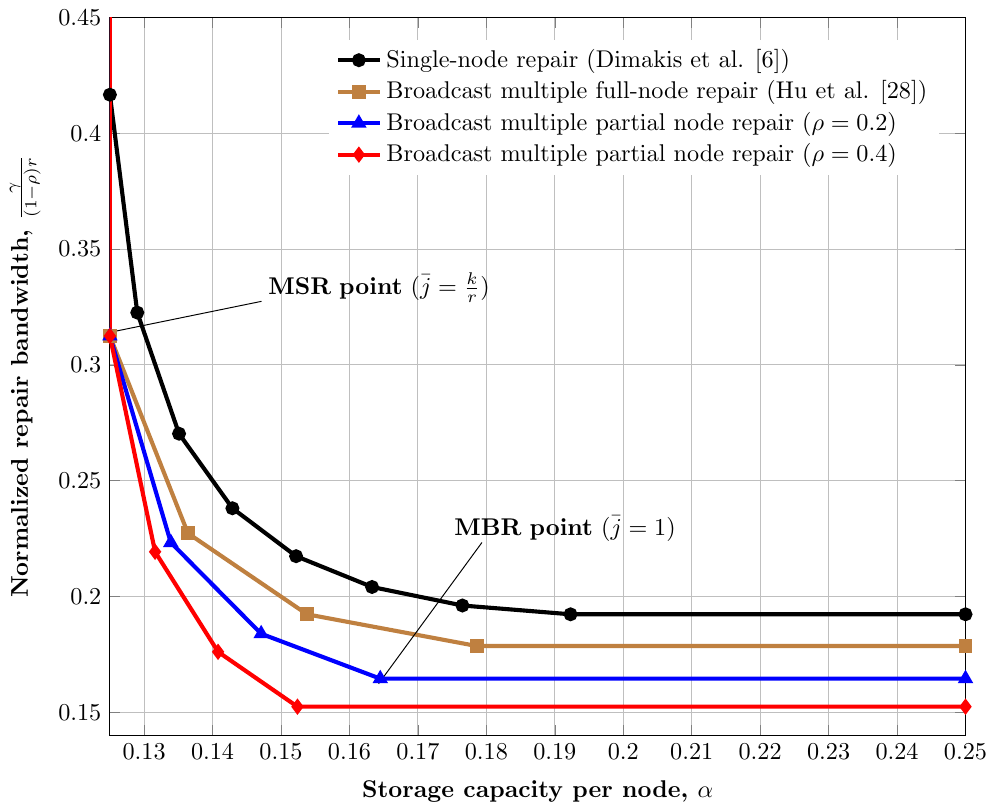}
    \caption{Trade-off curve between the repair bandwidth and storage, $M=1,k=8$ and $d=10$ helper nodes. For single node failure $r=1$ and for multiple node failures $r=2$.} \label{fig:tradeoff}
\end{figure}

\subsection{Optimal trade-off curve}
In Fig. \ref{fig:tradeoff}, we plot the storage vs. repair bandwidth trade-off for single-node repair \cite{5550492}, broadcast repair of multiple full node failures \cite{7459908}, and partial repair of multiple nodes. Fig. \ref{fig:tradeoff} illustrates that utilizing the unerased portion of data on a failed node reduces the repair bandwidth significantly. The repair bandwidth is normalized by the number of failed nodes, and the fraction of data that is erased, in order to make a fair comparison with single node repair and full node repair respectively. We call it the \textit{normalized repair bandwidth}. We observe that the MSR point does not improve over full node repair even with partial repair, that is, the unerased data in a faulty node does not help in the repair of the erased packets in any meaningful way for the MSR point. However, the normalized repair bandwidth decreases significantly at all other points on the trade-off curve.

\subsection{Verifying the reconstruction property}

We ran experiments using the \textit{sagemath} python library \footnote[1]{Code available at: \url{https://github.com/nitishmital/functional-repair}} to verify the preservation of the reconstruction property, that is, $\dim(W_{\mathcal{A}_{dc}}) \geq P^{*}$ over multiple repair rounds. We initially store $(d-(\bar{j}-1)r)(n-r)$ linearly independent evaluations of the encoding linearized polynomial on $n-r$ storage nodes, and use the repair scheme to populate the remaining $r$ nodes with content. We run $100$ repair rounds, where the indices of the failing nodes and the helper nodes are sampled uniform randomly from the $n$ nodes. We then compute the dimensions of the subspace obtained from sets of $k$ nodes sampled randomly in each of $50$ trials, and record the minimum dimension from the subspaces obtained from those sets of $k$ nodes, as well as the average dimension over those sets of $k$ nodes observed in the trials. It is observed that the value of $\dim(W_{\mathcal{A}_{dc}})$ decreases as the nodes go through repair rounds, and approaches $P^*$ asymptotically. Table \ref{results} records the results of experiments with multiple parameter combinations, for different points on the corresponding trade-off curves for those parameters, using the proposed scheme for full-node repair. We observe that $\dim(W_{\mathcal{A}_{dc}}) \geq P^{*}$ is satisfied for arbitrary parameter combinations if $e$ and $q$ are chosen appropriately.

\begin{table}[htb]
\begin{center}
\caption{Results of experiments with multiple parameter combinations verifying the reconstruction property.} \label{results}
 \begin{tabular}{| c c c c | c | c c c c c||} 
 \hline 
$n$ & $k$ & $d$ & $r$ & $\bar{j}$ & $q$ & $e$ & $P^{*}$ & \text{min} & avg \\ [0.5ex] 
 \hline\hline
 \multirow{3}{*}{ 27 } & \multirow{3}{*}{ 15 } & \multirow{3}{*}{ 17 } & \multirow{3}{*}{5} & 1 & 29 & 0 & 180 & 191 & 191 \\
& & & & 2 & 29 & 0 & 155 & 158 & 158 \\
& & & & 3 & 257 & 2 & 105 & 105 & 105\\
 \hline 
\multirow{4}{*}{ 24 } & \multirow{4}{*}{ 16 } & \multirow{4}{*}{ 16 } & \multirow{4}{*}{4} & 1 & 29 & 1 & 160 & 170 & 170 \\
 & & & & 2 & 29 & 1 & 144 & 149 & 149 \\
 & & & & 3 & 29 & 1 & 112 & 114 & 114 \\
 & & & & 4 & 29 & 0 & 64 & 64 & 64 \\
 \hline
\multirow{3}{*}{ 20 } & \multirow{3}{*}{ 12 } & \multirow{3}{*}{ 12 } & \multirow{3}{*}{4} & 1 & 29 & 1 & 96 & 96 & 96 \\
& & & & 2 & 29 & 1 & 80 & 84 & 84.98 \\
& & & & 3 & 29 & 0 & 48 & 48 & 48 \\
\hline

\multirow{4}{*}{ 16 } & \multirow{4}{*}{ 12 } & \multirow{4}{*}{ 12 } & \multirow{4}{*}{3} & 1 & 1021 & 3 & 90 & 91 & 91 \\
& & & & 2 & 1021 & 3 & 81 & 83 & 83 \\
 & & & & 3 & 257 & 3 & 63 & 63 & 63 \\
 & & & & 4 & 257 & 0 & 36 & 36 & 36 \\

\hline 
 
\multirow{4}{*}{ 16 } & \multirow{4}{*}{ 8 } & \multirow{4}{*}{ 11 } & \multirow{4}{*}{2} & 1 & 29 & 1 & 64 & 65 & 74.1 \\
 & & & & 2 & 29 & 1 & 60 & 64 & 68.84 \\
 & & & & 3 & 29 & 1 & 52 & 54 & 55.56 \\
 & & & & 4 & 29 & 1 & 40 & 40 & 40 \\
 \hline
\multirow{5}{*}{ 14 } & \multirow{5}{*}{ 10 } & \multirow{5}{*}{ 10 } & \multirow{5}{*}{2} & 1 & 29 & 2 & 60 & 60 & 60 \\
 & & & & 2 & 29 & 1 & 56 & 57 & 57.96 \\
 & & & & 3 & 29 & 2 & 48 & 49 & 49 \\
 & & & & 4 & 29 & 2 & 36 & 36 & 36.92 \\
 & & & & 5 & 127 & 0 & 20 & 20 & 20 \\
\hline
\multirow{2}{*}{ 9 } & \multirow{2}{*}{ 6 } & \multirow{2}{*}{ 6 } & \multirow{2}{*}{3} & 1 & 1021 & 3 & 27 & 27 & 27 \\
 & & & & 2 & 1021 & 0 & 18 & 18 & 18 \\
 \hline
\end{tabular}
\end{center}
\end{table}

\subsection{Subpacketization}

As described in Section \ref{sec:rw}, existing works mainly consider exact repair on the MSR and MBR points, but not the intermediate points on the trade-off curve. Existing constructions for repairing Reed Solomon (RS) codes that meet the cut-set bound employ an exponential in $n$ subpacketization \cite{10.1145/2897518.2897525,Tamo2017OptimalRO,8717612,7457282}. The construction in \cite{Tamo2017OptimalRO} achieves the cut-set bound for the MSR point but requires a subpacketization $S \approx n^n$. For combination no. 1 in Table \ref{results} , where $n=27$, the subpacketization required using the scheme in \cite{Tamo2017OptimalRO} is a practically infeasible $S \approx 27^{27}$. Among the works that propose repair schemes using non-RS codes, \cite{5961826} proposes a product-matrix construction that achieves the cut-set bound for a linearly scaling subpacketization level of $S=d-k+r$, but with the constraint $n \geq 2k-1$. The construction in \cite{7990181} achieves the MSR cut-set bound for general parameters with a subpacketization of $S \approx r^{\lceil \sfrac{n}{r} \rceil}$, which for the parameters in combination no. 1 in Table \ref{results}, gives $S = 5^6$. 

The proposed scheme achieves a subpacketization level that scales linearly with $n, k, d$ and $r$, for all points on the trade-off curve for full node repair. For example 1 in the above table, the subpacketization for the MSR point is $S=d-k+r = 7$, while that for the MBR point is $S = d = 17$. The subpacketization is given by $S=d-(\bar{j}-1)r$ for different points on the trade-off curve parameterized by $\bar{j}$. For partial repair, the subpacketization is given by $S=(d- (\bar{j}-1)r)\xi, \xi \in \mathbb{N}$, which is also linear in $n,k,d$ and $r$. 

\subsection{Input/Output cost}
The input-output cost is defined as the number of symbols that need to be read by a helper node from its memory, which are then linearly combined and sent to the newcomers. For the repair of $r$ newcomers, the proposed scheme achieves an input-output cost of $r+e$. As recorded in Table \ref{results}, small values of $e$ are often sufficient to achieve the cut-set bound, implying a low input-output cost of $r$.

\subsection{Computational complexity}
The finite field operations are in the finite extension field $\mathbb{F}(q^l)$. 
For reconstruction of the file from any $k$ nodes, the DC interpolates a linearized polynomial of degree $P$ in the finite field $\mathbb{F}_{q^l}$. The complexity of interpolation of a linearized polynomial is $\mathrm{O}(P^\epsilon)$ operations in $\mathbb{F}_q$, where $ \epsilon $ is the matrix multiplication exponent \cite{10.1145/3087604.3087617}. One of the frequently used fast algorithms for matrix multiplication is known as the Strassen algorithm \cite{10.1007/BF02165411}, which achieves the matrix multiplication exponent $\epsilon=2.807$. Therefore, the fastest reconstruction complexity is $\mathrm{O}(P^{2.807})$ operations in $\mathbb{F}_q$.

The repair complexity, defined as the computational complexity of repair operations, is smaller than that of random linear coding, because only rows of $\bar{j}r < dr$ packets are multiplied with the parity matrix $\mathbf{M}_{\bar{j}r \times r}$ in the repair process, which makes the computations faster, unlike random linear coding where, in general, $d-(\bar{j}-1)r$ random linear combinations of $dr$ packets are computed during the repair process, requiring $dr(d-(\bar{j}-1)r)$ finite field operations. In the proposed scheme, the number of finite field operations during repair is given by $\bar{j}r(d-(\bar{j}-1)r)$. For the MBR point, where $\bar{j} = 1$, this results in a reduction in repair complexity by a factor of $d$. Since the number of coefficients of linear combination needed to be communicated is also reduced, the overhead is lower.

\section{Conclusions} \label{sec:conc}
In this paper, we studied large scale distributed storage systems and considered the problem of repair of partial failures of multiple nodes via broadcast transmissions over a wireless medium. We derive the optimal storage-repair bandwidth trade-off curve by constructing an information flow graph to represent the evolution of the system with time, and computing the minimum cut-set capacity in the information flow graph. It has been shown in previous literature that, compared to single node repair, repairing multiple nodes at once and exploiting the broadcast nature of the medium reduce the repair bandwidth per failed node. We illustrate that the optimal repair bandwidth is reduced even further by utilizing the remaining content in the nodes that experience partial failure. We derive the invariant conditions related to the way the subspaces stored by different nodes intersect, that are sufficient for the existence of a feasible functional regenerating code, and provide an intuitive insight into how functional regenerating codes may be constructed. We then present an explicit storage and repair framework for the functional repair of multiple node failures in a broadcast setting, achieving all the points on the trade-off curve, as illustrated in Fig. \ref{fig:tradeoff}, with high probability. We also extend the framework to the case when there is only partial failure of multiple nodes. The proposed storage and repair framework achieves multiple desirable characteristics for regenerating codes. These characteristics include achieving the optimal storage-repair bandwidth trade-off with high probability, for \textit{many feasible parameters} ($n,k,d,r$) not achieved in existing literature; achieving \textit{all points on the trade-off curve}; a subpacketization level that scales linearly with respect to the code parameters; low input-output cost; and low computation complexity during repair.

An interesting future research direction is the consideration of more realistic heterogeneous scenarios, in which the storage nodes have unequal capacities and experience unequal partial failures, as explored in \cite{benerjee2015tradeoff} for a flexible reconstruction degree, where every node in the system has a dynamic repair bandwidth and dynamic storage capacity. In this case, finding the min-cut capacity of the information flow graph must be formulated as a linear programming problem. While trivial extensions of the proposed scheme for the homogeneous setting in this paper can allow us to obtain a sub-optimal achievable scheme, a thorough analysis should potentially provide significant gains and interesting insights. It may be possible to group nodes together or assign different tolerances to partial failures to different nodes based on their storage capacities or connectivity to other storage nodes.
\bibliography{main}
\bibliographystyle{IEEEtran}

\section{Appendix}

\subsection{Proof of Theorem \ref{theorem 1}} \label{appendix1}

\begin{proof}
Consider an information flow graph $G$ that enumerates all possible failure/repair patterns and all possible DCs when the number of failures/repairs is bounded by $r$. We analyze the connectivity in the information flow graph to find the minimum repair bandwidth. Initially, the source delivers $\alpha$ bits each to $n$ nodes, which then become active while the source node becomes inactive. When $r$ nodes lose part of their data, a repair round is triggered in which they connect to $d$ surviving nodes and receive $\beta$ bits from each of them. Using the received messages, and the un-erased content in their local memories, the $r$ nodes recover their lost content. The active nodes, before the $s$-th repair round is triggered, are labelled as $\mathcal{R}_s^{act}\triangleq\{(s-1)n+1, \ldots , sn\}$. Consider that in the $s$-th  repair round, the nodes $\mathcal{R}_s^{f}\triangleq \{(s-1)n+(s-1)r+1, \ldots ,(s-1)n+sr\}$ are repaired. Denote the set of labels of the newcomers in the $s$-th round by $\mathcal{R}_s^{new}\triangleq\{ sn+(s-1)r+1, \ldots, sn+sr \}$, which represent the repaired nodes. The complete nodes are copied into the next round and labeled as  $\mathcal{R}_s^{comp}\triangleq \{\ i: i\in [sn+1:(s+1)n]\setminus \mathcal{R}_s^{new}\}$. The newcomers and the copied complete nodes together form the set of active nodes for the next repair round, i.e., $\mathcal{R}_{s+1}^{act}= \mathcal{R}_s^{new} \cup \mathcal{R}_s^{comp}$, while the nodes from all the previous rounds become inactive.

For the reconstruction property to hold, any DC that connects to the ``out-nodes'' of any $k$ active nodes must satisfy
\begin{align}\label{bound}
    C&=\text{mincut}(S,DC) \\
    &\geq \sum_{s=1}^{\sfrac{k}{r}}\min \{ (r\alpha_1 + (d-r(s-1))\beta, r\alpha) \}.
\end{align}
First, we show that there exists an information flow graph $G^\prime$, for which Eq. \eqref{bound} holds with equality. Consider that after the storage nodes have gone through $h$ repair rounds, a DC connects to the nodes with indices $\mathcal{R}_{h}^{new}$. Consider a cut $(U,\bar{U})$ between $S$ and $DC$, which separates the graph into the disjoint sets of nodes $U$ and $\bar{U}$, constructed as follows. For the $s$-th repair round, if $r\alpha \leq r\alpha_1 + (d-(s-1)r)\beta$, then we include the nodes $x_{in}^{\mathcal{R}_s^{new}}$ in $U$, and $x_{out}^{\mathcal{R}_s^{new}}$ in $\bar{U}$, similarly to the cut $\chi_2$ in Fig. \ref{fig:flow1}; otherwise, we include $x_{in}^{\mathcal{R}_s^{new}},x_{out}^{\mathcal{R}_s^{new}},x_{out}^{\mathcal{R}_s^{f}}$, and all auxiliary nodes in $\bar{U}$, while the nodes $x_{mid}^{\mathcal{R}_s^{f}}$ and $x_{in}^{\mathcal{R}_s^{act}}$ are included in $U$, similarly to the cut $\chi_1$ in Fig. \ref{fig:flow1}. We argue that the capacity of the cut $(U,\bar{U})$ meets that of Eq. \eqref{bound} with equality.

Second, we argue that any information flow graph has at least the cut capacity in Eq. \eqref{bound}. We note that there is a topological order of the nodes in an information flow graph by which any node $\nu_i$ having incoming edges only from nodes in $\bar{U}$, also belongs to $\bar{U}$, and an edge from $\nu_i$ to $\nu_j$ implies $i<j$. The min-cut in a repair round can be of only two types. In a ``type-1'' cut, all helper nodes not yet included in $\bar{U}$, and the out-vertices of the faulty nodes, are included in $\bar{U}$, as illustrated by the cut $\chi_1$ in Fig. \ref{fig:flow1}. A type-1 cut includes both the in-vertices and the out-vertices of the newcomer nodes in $\bar{U}$. In a ``type-2'' cut, all helper nodes and the in-vertices of the newcomers are included in $U$, while the out-vertices of the newcomers are included in $\bar{U}$, as illustrated by the cut $\chi_2$ in Fig. \ref{fig:flow1}. Other cuts that include the in-vertices of a subset of the newcomers in $U$, while including the in-vertices of the remaining newcomers in $U$, always have a capacity larger than type-1 cuts. 

We illustrate this using Fig. \ref{fig:flow1}: Suppose cut $\chi_3$ passes through the edge $x_{in}^{5}\rightarrow x_{out}^5$, and the edges $x_{mid}^2 \rightarrow x_{out}^2, x_{out}^{3}\rightarrow h^3, x_{out}^{4} \rightarrow h^4$. The capacity of cut $\chi_3$ is given by $\alpha_1 + 2\beta + \alpha$, which is always greater than the capacity of the cut $\chi_1$, since $\alpha_1 + 2\beta + \alpha > 2\alpha + 2\beta$. 
Therefore, we see that the minimum cut is always either type-1 or type-2 in any particular repair round, and the minimum cut thus obtained achieves the cut capacity given by Eq. \eqref{bound}. 

The expression for the min-cut capacity is derived in the following way. The contribution of the first repair round to the minimum cut capacity is given by $\min \{ r\alpha_1 + d\beta , r\alpha \}$, where the first term denotes the capacity contribution from a type-1 cut, and the second term denotes the contribution from a type-2 cut. Consider the second repair round. The $r$ nodes which are repaired in the first repair round already lie in $\bar{U}$, so edges originating from these nodes do not contribute to the min-cut capacity from the second round onwards. 

The min-cut capacity contribution by the second repair round is given by $\min \{ r\alpha_1 + (d-r)\beta , r\alpha \}$. We follow this procedure of passing the min-cut through each repair round with a type-1 or a type-2 cut until the $DC$ lies in $\bar{U}$, which happens when the $k$ nodes to which the $DC$ is connected to lie in $\bar{U}$. When we sum the contributions from each repair round to the min-cut capacity, we obtain Eq. \eqref{bound}.

From Proposition \ref{prop:1}, the min-cut capacity must be greater than the file size to ensure that the DC is able to reconstruct the file from any $k$ nodes. Therefore, the following must be satisfied for guaranteed file reconstruction:
\begin{align} \label{sufficiency}
    \sum_{s=1}^{\sfrac{k}{r}}\min \{ (r\alpha_1 + (d-r(s-1))\beta, r\alpha) \} \geq M.
\end{align}

We are interested in characterizing the achievable trade-offs between the storage $\alpha$ and the repair bandwidth $d\beta$ for given $(n,k,\rho)$. If $r\alpha \leq r\alpha_1 + (d-k+r)\beta$, then the min-cut is type-2 in each of the $\sfrac{k}{r}$ repair rounds; if $ r\alpha_1 + (d-k+r)\beta \leq r\alpha \leq r\alpha_1 + (d-k+2r)\beta$, then the min-cut is type-2 for the first $\sfrac{k}{r}-1$ repair rounds, but type-1 in the $\sfrac{k}{r}$-th repair round. In general, if $r\alpha_1 + (d-rs)\beta \leq r\alpha \leq r\alpha_1 + (d-r(s-1))\beta, s \in [\sfrac{k}{r}]$, then the min-cut is type-2 for the first $s$ repair rounds, and type-1 for the remaining $\sfrac{k}{r}-s$ repair rounds. \\

Let $b_{s-1} \triangleq \frac{\frac{d-k}{r}+s}{1-\rho}\beta, s=[\sfrac{k}{r}]$. The capacity of the min-cut is a piecewise-linear function of $\alpha$ given by
\begin{align}
    C(\alpha)&=\left\{
                \begin{array}{ll}
                k\alpha, \hspace{80pt} \alpha \in (0,b_0]\\
                (k-r)\alpha + \left( r\alpha_1 + (d-k+r)\beta\right),\\
                \hspace{100pt} \alpha \in (b_0,b_1]\\
                \vdots \\
                r\alpha + \sum_{i=1}^{\sfrac{k}{r}-1} \left(  r\alpha_1 + (d-k+ir)\beta\right),\\ \hspace{100pt} \alpha \in (b_{\sfrac{k}{r}-2}, b_{\sfrac{k}{r}-1}]\\
                 \sum_{i=1}^{\sfrac{k}{r}} \left(  r\alpha_1 + (d-k+ir)\beta\right),\\
                 \hspace{100pt} \alpha \in (b_{\sfrac{k}{r}-1}, \infty]
                \end{array}
                \right. \\
            &= \left\{
                \begin{array}{ll}
                k\alpha, \hspace{3.5pt} \alpha \in (0,b_0]\\
                (k-ir(1-\rho))\alpha +(1-\rho) \sum_{j=0}^{i-1} rb_j,\\
                \hspace{20pt} \alpha \in (b_{i-1},b_i], i=1,2,\ldots, \sfrac{k}{r}-1 \\
                k\rho \alpha + (1-\rho)\sum_{j=0}^{\sfrac{k}{r}-1} rb_j, \\ \hspace{20pt} \alpha \in (b_{\sfrac{k}{r}-1},\infty]
                \end{array}
                \right.    
\end{align}
Note that $C(\alpha)$ is a strictly increasing function. To find the minimum $\alpha$ for a given repair bandwidth $\gamma=d\beta$ such that $C(\alpha)\geq M$, we let $\alpha^* = C^{-1}(M)$ to obtain
\begin{align}
    \alpha^*&=\left\{
                \begin{array}{ll}
                \frac{M}{k} \hspace{45pt}  M\in (0,kb_0] \\
                \frac{M - g(i)\gamma}{k-ir(1-\rho)} \hspace{15pt} M \in \Big[(k-ir(1-\rho))b_{i-1} \\
                +(1-\rho)\sum_{j=0}^{i-1}rb_j ,(k-ir(1-\rho))b_{i} \\
                \hspace{2cm} +(1-\rho)\sum_{j=0}^{i-1}rb_j\Big]
                \end{array}
                \right. \\
             &=\left\{
                \begin{array}{ll}
                  \frac{M}{k} \hspace{16.5mm} \gamma \in \left[ f(0), \infty \right) \\
                  \frac{M - g(i)\gamma}{k-ir(1-\rho)} \hspace{6mm} \gamma \in \left[ f(i), f(i-1)\right] 
                \end{array} \label{curve_alpha}
              \right. 
\end{align}
\end{proof}
\subsection{Proof of Theorem \ref{theorem 2}} \label{appendix2}

\begin{proof}
The proof follows essentially the same steps as in the proof for Theorem \ref{theorem 1}. $k$ nodes are divided into $p$ groups of $r$ nodes where $p=\lfloor \frac{k}{r} \rfloor$, and a remaining group of $k-k_0$ nodes. The min-cut capacity contribution by the group of $k-k_0$ nodes is given by $\min \{ (k-k_0) \alpha , (k - k_0)\alpha_1 + (d-k_0)\beta \}$, while the min-cut capacity contribution of the other groups of $r$ nodes is computed in exactly the same manner as for Theorem \ref{theorem 1}. The min-cut capacity is therefore written as:
\begin{align} \label{sufficiency_rnotk}
  \nonumber &C(\alpha) =  \min \Big\{ (k-k_0) \alpha , (k - k_0)\alpha_1 + (d-k_0)\beta \}\\
   & \hspace{0.4cm} + \sum_{s=1}^{\sfrac{k}{r}}\min \{ (r\alpha_1 + (d-r(s-1))\beta, r\alpha) \Big\}.
\end{align}
The rest of the derivation of the piecewise linear function follows the same procedure as in the proof for Theorem \ref{theorem 1}.
\end{proof}

\subsection{Proof of Lemma \ref{lemma:1}} \label{appendix3}
\begin{proof}
We have
\begin{align}
    &\dim\Big(W_A \cap  W_{\mathcal{B}}\Big) = \dim\Big(W_A\cap \sum_{i=1}^{v}W_{i}\Big)\\
    & = \dim\left(W_A \cap \left( \sum_{t=1}^{\lfloor \sfrac{v}{r} \rfloor}W_{\mathcal{R}_t} + W_{\mathcal{R}'}\right)\right)\\
\nonumber  &= \sum_{u=1}^{\bar{j}-1} \cancelto{0}{\dim\Big(W_A \cap W_{\mathcal{R}_{u}}\Big\vert \sum_{t=1}^{u-1}W_{\mathcal{R}_{t}}\Big)}\\
 & \hspace{0.6cm} +\dim\Big(W_A \cap W_{\mathcal{B}\setminus (\mathcal{R}_1,\ldots,\mathcal{R}_{\bar{j}-1})}\Big\vert \sum_{t=1}^{\bar{j}-1}W_{\mathcal{R}_{t}}\Big) \label{chain_rule}\\ 
 &\overset{a}{=} \dim\Big(W_A \cap W_{\mathcal{B}\setminus (\mathcal{R}_1,\ldots,\mathcal{R}_{\bar{j}-1})}\Big\vert \sum_{t=1}^{\bar{j}-1}W_{\mathcal{R}_{t}}\Big)\\
 \nonumber &\overset{b}{=} \dim\Big(W_A \cap W_{\mathcal{R}_{\bar{j}}} \Big\vert \sum_{t=1}^{\bar{j}-1}W_{\mathcal{R}_t} \Big) + \\
 & \hspace{0.7cm} \dim\Big(W_A \cap W_{\mathcal{B}\setminus (\mathcal{R}_1,\ldots,\mathcal{R}_{\bar{j}})}  \Big\vert \sum_{t=1}^{\bar{j}-1}W_{\mathcal{R}_t} \Big) \label{step b}
    \end{align}
    where step $a$ is due to property \textbf{L1}, while step $b$ is due to \textbf{L2}. 
    Using Eq. \eqref{step b} recursively, we obtain
    \begin{align}
        \nonumber &\dim\Big(W_A \cap W_{\mathcal{B}}\Big) = 
        \sum_{s\geq \bar{j}}^{\lfloor\sfrac{v}{r}\rfloor} \dim\Big(W_A \cap W_{\mathcal{R}_s}\Big\vert \sum_{t=1}^{\bar{j}-1}W_{\mathcal{R}_t}\Big) 
         \\
        & \hspace{2cm} +\dim\Big(W_A\cap W_{\mathcal{R}'}\Big\vert \sum_{t=1}^{\bar{j}-1}W_{\mathcal{R}_t}\Big) \label{decomposition} \\
        & \hspace{2cm}\overset{c}{=} \sum_{s\geq \bar{j}}^{\lfloor\sfrac{v}{r}\rfloor} \dim\Big(W_A \cap W_{\mathcal{R}_s}\Big\vert \sum_{t=1}^{\bar{j}-1}W_{\mathcal{R}_t}\Big)
    \end{align}
    where step $c$ follows from $\textbf{L1}$.
\end{proof}

\end{document}